\documentclass[a4paper,12pt]{article}
\usepackage{latexsym}
\usepackage{graphics} 
\usepackage{epsfig} 
\usepackage{amsmath,amssymb,latexsym,amsfonts,amsthm}
\usepackage{enumerate}
\usepackage[longnamesfirst]{natbib}

\newtheorem{thm}{Theorem}[section]
\newtheorem{prop}[thm]{Proposition}

\newtheorem{remark}[thm]{Remark}

\newtheorem{met}[]{Method}

\newcommand{\bE}{\ensuremath{\mathbf{E}}}
\newcommand{\bP}{\ensuremath{\mathbf{P}}}
\newcommand{\bR}{\ensuremath{\mathbf{R}}}
\newcommand{\bQ}{\ensuremath{\mathbf{Q}}}

\newcommand{\cC}{\ensuremath{\mathcal{C}}}
\newcommand{\cF}{\ensuremath{\mathcal{F}}}

\begin{document}
\title{\LARGE{\bf A Numerical Scheme Based on Semi-Static Hedging Strategy}}

\author{Yuri Imamura\footnote{Corresponding author: yuri.imamura@gmail.com}, Yuta Ishigaki, \\
Takuya Kawagoe and Toshiki Okumura 
}

\maketitle
\thispagestyle{empty}

\abstract{
In the present paper, we introduce a 
numerical scheme for the price of 
a barrier option when the price of the underlying follows a diffusion process. 
The numerical scheme is based 
on an extension of a static hedging formula
of barrier options. For getting the static hedging formula, 
the underlying process needs to have a symmetry. 
We introduce a way to 
``symmetrize" a given diffusion process. 
Then the pricing of a barrier option is reduced to that 
of plain options under the symmetrized process. 
To show how our symmetrization scheme works, 
we will present 
some numerical results applying (path-independent) 
Euler-Maruyama approximation to 
our scheme, comparing them with the path-dependent Euler-Maruyama scheme 
when the model is of the Black-Scholes, 
CEV, Heston, and $ (\lambda) $-SABR, respectively.
The results show the effectiveness of our scheme.  
\\}

\ 


\section{Introduction}
In financial practice, the pricing (and hedging) of 
barrier type derivatives becomes more and more important.  
In the Black-Scholes environment, 
some analytic formulas are available in 
\citet{M}. 

If the underlying process is a diffusion process which is more complicated 
than a Geometric Brownian Motion, 
it will not be able, basically, to rely on anymore an analytic formula. 
Instead one should resort to some numerical analysis. 
There is a problem, however, arising from the 
path-dependence of the pay-off function. 
\citet{G} pointed out that the weak convergence rate against 
the time-discretization gets worse compared with the 
standard path-independent pay-off cases
due to the failure in the observation of 
hitting between two time steps. 
He showed that the weak order 
of Euler-Maruyama approximation 
is $\frac{1}{2}$, 
which is much slower than the standard case where the order is $1$.  

In the present paper, we will introduce a new numerical scheme 
where the pricing (and hedging) of barrier options are reduced to 
that of plain (path-independent) ones. 
The scheme is based on an observation made by \citet{CL} 
which we will refer to as ``arithmetic put-call symmetry" (APCS). 
In the Black-Scholes economy, it is well-recognized that 
the reduction 
is possible due to the reflection principle 
(see \citet{Imamura}). 
The put-call symmetry is an extension of the reflection principle, 
with which a semi-static hedge is still possible. 

There are two keys in our scheme; 
\begin{enumerate}
\item 
For a given diffusion $ X $ and a real number $ K $, we can 
find another diffusion $ \tilde{X} $ which satisfies the 
APCS 
at $ K $ (see section \ref{main}). 
We call this procedure ``symmetrization".
\item 
For $ T> 0 $, 
the expectations $ E [f(X_T)1_{\{\tau >T\} }] $ and 
$ E [f(\tilde{X}_T)1_{\{\tilde{\tau} >T \}}] $
coincides, where $ \tau $ and $ \tilde{\tau}$ are
the first hitting time at $ K $ of $ X $ and $ \tilde{X} $, respectively. 
\end{enumerate}
We do not anymore 
regard the equation for semi-static hedging
but just a relation to calculate the expectation of 
the diffusion with a Dirichlet boundary condition in terms of 
those without boundary conditions. 
In other words, the pricing and hedging 
are reduced to path-independent ones, where many stable 
techniques are available. In this paper, we will present 
some numerical results of applying (path-independent) 
Euler-Maruyama (EM) approximation to 
our scheme, comparing them with the path-dependent EM 
under 
Constant Elasticity of Volatility (CEV) models (\citet{Cox}) 
including as a special case the Black-Scholes (BS) model, 
and stochastic volatility models of Heston's (\citet{nimalin} ) and 
$(\lambda) $-SABR (\citet{Hagan,laurence}). 

This paper consists of two parts. In the first part, 
the discussion of our new scheme is concentrated on one-dimensional 
diffusion models, while the latter part deals with 
applications to the stochastic volatility models. 
Mathematically, the first part is somehow {\em self-contained}, while
one may think 
the latter part to be dependent on the result in \citet{AI}. 
The fact is that we have found, 
in advance of \citet{AI}, 
through numerical experiments 
how it should be applied to 
stochastic volatility models (see \citet{IIO}). 
In anyway, the main aim of the present paper is 
to introduce the new scheme and to report numerical results
which show the effectiveness of the scheme. 
In order to ensure the consistency of the experiments, 
we present detailed descriptions. 

The paper is organized as follows. 
In Section \ref{PCS}, we recall the APCS and 
how it is applied to the pricing and semi-static hedging of barrier options. 
In section \ref{APCS-dif} 
we give a sufficient condition shown by 
\citet{CL} under which APCS holds. 
In Section \ref{main}, 
we introduce a way to 
``symmetrize" a given diffusion process. 
We then show that by using the symmetrized process which satisfies APCS, 
the pricing of barrier option is reduced to that 
of two plain options. 
In section \ref{sec-numerical}, 
we give numerical examples 
under our symmetrized approximation method. 
The results of the path-wise EM scheme (in section \ref{sec-EM}) 
and our new scheme are given when 
the underlying asset price process 
follows CEV with the volatility elasticity $ \beta = 1$
which is nothing but a BS model and other elasticities. 
From Section \ref{sec-svm}, we discuss applications to 
stochastic volatility models. 
The symmetrized method is also applicable to the stochastic volatility models
where the underlying price process and its volatility follows a (degenerate)
$2$-dimensional diffusion process. 
In section \ref{num;sv;single}, 
we give numerical results under Heston model and 
$\lambda$-SABR model. 
In section \ref{sec-double}, we show that the symmetrization scheme 
also works for the 
pricing of double barrier option.  
One will find that our scheme overwhelms 
the path-wise EM in all numerical results. 

\subsubsection*{Acknowledgments}
The authors are grateful to Professors Jir\^o Akahori,  
Arturo Kohatsu-Higa, and Hideyuki Tanaka 
for many valuable comments and 
for careful reading the
manuscript and suggesting several improvements. 

\section{The Put-Call Symmetry Method 
for One Dimensional Diffusions}\label{PCS}
\subsection{Arithmetic Put-Call Symmetry}\label{subsec-APCS}
Let $ X $ be a real valued diffusion process defined on 
a complete filtered probability space
$ (\Omega,\mathcal{F}, \mathbf{P},\{ \mathcal{F}_t\} ) $ 
which satisfies the usual conditions. 
For fixed $K >0$, 
we say that 
{\em arithmetic put-call symmetry} (APCS) at $ K $  holds for $ X $ 
if the following equation is satisfied ; 
\begin{equation*}
\mathbf{E}[G(X_t - K)\bigm| X_0=K]
=\mathbf{E}[ G(K- X_t )\bigm| X_0=K],
\end{equation*}
for 
any bounded measurable function $ G $ and $t \geq 0$. 
The APCS at $K$ is alternatively defined to be 
the equivalence in law between 
$X_t - X_{t \wedge \tau}$ and $X_{t \wedge \tau} - X_t$ for any $t \geq 0$,
where $ \tau $ is the first hitting time of $ X $ to $ K $.  

Intuitively, the APCS means the following. 
For every path of $X$ which crosses the level $K$ and is found at time 
$t$ at a point below $K$, there is a ``shadow path" obtained from 
the reflection with respect to the level $K$ 
which exceeds this level at time $t$, and these two paths 
have the same probability. 
For one-dimensional Brownian motion, APCS holds for any $K>0$ 
since the reflection principle holds. 

In \citet{CL}, the APCS, or more precisely, {\em PCS}\footnote{
They defined PCS as the 
equality of the distribution between $\frac{X_T}{X_0}$ under 
$\bP$ and $\frac{X_0}{X_T}$ under $\bQ$, 
where $\frac{\bQ}{\bP} = \frac{X_T}{X_0}$.},
is applied to the pricing and
{\it semi-static} hedging of barrier options. 
Semi-static hedging means replication of the barrier contract 
by trading European-style claims at no more than two times after inception. 

In more detail, we have the following;
if $ X $ satisfies APCS at $ K$, 
then for any bounded measurable $ f $ and $ T > 0 $, 
\begin{equation*}
\begin{split}
\mathbf{E}[f({X}_T)I_{\{\tau >T\}}]
&=\mathbf{E}[f({X}_T)I_{\{X_T > K,\ \tau >T \}}]\\
&=\mathbf{E}[f({X}_T)I_{\{X_T>K\}}] 
- \mathbf{E}[f({X}_T)I_{\{X_T>K,\ \tau \leq T \}}],\\
\end{split}
\end{equation*}
where 
\begin{equation}\label{tau}
\tau = \inf \{t \geq 0 : X_T \leq K\}.
\end{equation} 
By APCS of $X$, we see that 
\begin{equation*}
\begin{split}
\mathbf{E}[f({X}_T)I_{\{X_T>K,\ \tau \leq T \}}]
&=\bE[\bE [f({X}_T)I_{\{X_T>K\}} | \cF_{\tau}]I_{\{\tau \leq T \}} ]\\
&=\bE[\bE [f(2K-{X}_T)I_{\{X_T<K\}} | \cF_{\tau}]I_{\{\tau \leq T \}} ].
\end{split}
\end{equation*}
Hence we obtain the following equation;
\begin{equation}\label{PCSM}
\begin{split}
\mathbf{E}[f({X}_T)I_{\{\tau >T\}}]
&=\mathbf{E}[f({X}_T)I_{\{X_T>K\}}]\\
&\qquad -\mathbf{E}[f(2K-{X}_T)I_{\{X_T<K\}}].
\end{split}
\end{equation}
Of the equation (\ref{PCSM}), 
the left-hand-side reads the price of a barrier option 
written on $X$, whose pay-off is $f$, 
knocked out at $ K $, and the right-hand-side is the price of
a combination of two plain-vanilla options. 

Here is a description of the hedging strategy of a barrier option 
implied from the right-hand-side of (\ref{PCSM}); 
\begin{itemize}
\item[1.]
Hold a plain-vanilla options whose pay-off is $f(X_T)$ 
if the price at the maturity is less than $K$, and is nothing 
if the price at the maturity is greater than $K$. 
Moreover 
short-sell a plain-vanilla options whose pay-off is 
$f(2K-X_T)$ if the price at the maturity is greater than $K$, 
and is nothing if the price at the maturity is less than $K$. 
\item[2.]
Keep the above position until the price hits the barrier $K$. 
If the price never hits $K$ until the maturity, 
the pay-off is $f(X_T)$. 
\item[3.]
If the price hits $K$, clear both plain-vanilla 
positions at the hitting time. Indeed, 
the value of two options are exactly the same at $\tau$. 
\end{itemize}
\subsection{APCS of diffusion process}\label{APCS-dif}

Let $X$ be a solution to 
the following one-dimensional 
stochastic differential equation (SDE) 
driven by a Brownian motion $W$ in $\bR$, 
\begin{equation}\label{diffusion}
dX_t = \sigma (X_t)dW_t + \mu (X_t)dt. 
\end{equation}
Here we assume the following hypotheses; 
\begin{description}
\item[(H1)]
$\sigma: \bR \rightarrow \bR$ 
and $\mu: \bR \rightarrow \bR$ 
are locally bounded measurable functions 
such that 
the linear growth condition is satisfied, ie, 
for a  constant $C$, $|\sigma(x)| + | \mu (x)| \leq C( 1 + |x|)$ for any $x \in \bR$.
\item[(H2)]
The following condition is satisfied; \\
\begin{equation*}
\sigma (y) \neq 0 
\Longleftrightarrow 
\sigma^{-2} \mbox{ is integrable in a neighborhood of } y.  
\end{equation*}
\end{description}

Then we have the following result on the 
uniqueness of the solution to (\ref{diffusion}); 
\begin{thm}[Theorem $4$, \citet{ES}]\label{Thm-ES}
Under {\bf  (H2)}, there exists a unique (in law) solution 
satisfying SDE (\ref{diffusion}). 
\end{thm}
Moreover, by the linear growth condition {\bf (H1)}, 
the unique (in law) solution will not explode in finite time. 

\citet{CL} 
gave a sufficient condition for a solution of 
(\ref{diffusion}) to satisfy PCS at $0 \in \bR$. 
The following Proposition 
is essentially a 
corollary to Theorem $3.1$ in \citet{CL}. 
\begin{prop}\label{prop-CL}
If the coefficients further satisfy the following conditions;  
\begin{equation}\label{sigma}
\sigma (x) = \varepsilon (x) \sigma (2K-x)\ (x \in \bR\setminus \{K\}),
\end{equation}
for a measurable $\varepsilon : \bR \rightarrow \{-1,1\}$ and
\begin{equation}\label{mu}
\mu(x)=-\mu(2K-x) \ (x \in \bR\setminus \{K\}),
\end{equation}
then APCS at $ K $ holds for $ X $.
\end{prop}
\begin{proof}
By the uniqueness in law, 
it is sufficient to show that 
$(X_{t \wedge \tau} - (X_t - X_{t \wedge \tau}) )_{t \geq 0}$ 
solves the SDE (\ref{diffusion}). 
By the assumptions $(\ref{sigma})$ and $(\ref{mu})$, we obtain that 
\begin{eqnarray*}
\begin{split}
X_{t \wedge \tau} - (X_t - X_{t \wedge \tau}) 
&=X_{t \wedge \tau} -\int_{t \wedge \tau}^t
\sigma(X_s)dW_s-\int_{t \wedge \tau}^t\mu(X_s)ds\\
&=X_{t \wedge \tau}
-\int_{t \wedge \tau}^t\varepsilon (X_s) \sigma(2K-X_s)dW_s\\ 
& \qquad 
+\int_{t \wedge \tau}^t \mu (2K-X_s) ds\\
&=X_{t \wedge \tau}
-\int_{t \wedge \tau}^t\varepsilon (X_s)\sigma(X_{\tau}-(X_s- X_{\tau}))dW_s\\ 
& \qquad 
+\int_{t \wedge \tau}^t \mu (X_{\tau}-(X_s- X_{\tau})) ds.
\end{split}
\end{eqnarray*}
We set $W'_t= W_{t\wedge \tau }-
\int_{t\wedge \tau }^t\varepsilon (X_s)d W_s$. 
Since
we obtain that  
\begin{eqnarray*}
\begin{split}
\langle W' \rangle (t) &= \langle W \rangle (t)\\
&=t,
\end{split}
\end{eqnarray*}
$W'$ is a Brownian motion (cf. \citet{IW} Chapter II, Theorem $6.1$.).  
Therefore we see that 
\begin{eqnarray*}
\begin{split}
X_{t \wedge \tau} - (X_t - X_{t \wedge \tau}) 
&=X_{0}+
\int_{0}^t\sigma(X_{s \wedge \tau} - (X_s - X_{s \wedge \tau}))dW_s'\\ 
&\qquad + \int_{0}^t\mu (X_{s \wedge \tau} - (X_s - X_{s \wedge \tau})) ds.
\end{split}
\end{eqnarray*}
Hence APCS at $ K $ holds.
\end{proof}

\subsection{Symmetrization of Diffusion Processes}\label{main}
We introduce a way to ``symmetrize" a given diffusion to satisfy APCS. 
By using this symmetrized process satisfying APCS, 
the pricing of barrier options is reduced to that of plain options. 

We start with a diffusion process $ X $ given 
as a unique solution of SDE $(\ref{diffusion})$. 
We do not assume that the 
coefficients have the symmetric conditions 
(\ref{sigma}) and (\ref{mu}). 
We then construct another diffusion $ \tilde{X} $ that satisfies
APCS at $ K $ in the following way. 
Put 
\begin{equation}\label{sigma2}
\tilde{\sigma} (x) := 
\begin{cases}
\sigma (x) &  x > K \\
\sigma (2K-x) & x \leq K,
\end{cases} 
\end{equation}
\begin{equation}\label{mu2}
\tilde{\mu} (x) := 
\begin{cases}
{\mu} (x) &  x > K \\
-{\mu} (2K-x)  & x \leq K,
\end{cases} 
\end{equation}
and consider the following SDE; 
\begin{equation}\label{tilde}
d\tilde{X}_t=\tilde{\sigma}(\tilde{X}_t)dW_t+\tilde{\mu}(\tilde{X}_t)dt.
\end{equation} 
Again by Theorem \ref{Thm-ES}, 
there is a unique (in law) solution $\tilde{X}_t$. 
Then we obtain the following result;
\begin{thm}\label{Thm}
It holds that 
\begin{equation}\label{new}
\begin{split}
\mathbf{E}[f({X}_T)I_{\{\tau >T\}}] 
&= \mathbf{E}[f(\tilde{X}_T)I_{\{\tilde{X}_T>K\}}]\\
& \qquad-\mathbf{E}[f(2K-\tilde{X}_T)I_{\{\tilde{X}_T < K\}}]. 
\end{split}
\end{equation}
\end{thm}
\begin{proof}
Since $\tilde{\sigma}$ and $\tilde{\mu}$ satisfy the condition 
$(\ref{sigma})$ and $(\ref{mu})$, 
APCS at $ K $ holds for $\tilde{X} $  
by Proposition \ref{prop-CL}. 
Then the equation (\ref{PCSM}) is valid for $\tilde{X}$. 
Moreover, by the definition of 
$\tilde{\sigma}$ and $\tilde{\mu}$, 
we have $\sigma(x) = \tilde{\sigma}(x)$ and 
$\mu (x)= \tilde{\mu}(x)$ for $x <K$. Therefore by the uniqueness in law of the
SDE, we have that 
$\{X_t\}_{t \leq \tau}=\{\tilde{X}_t\}_{t \leq \tau}$ pathwisely. 
Then we see that 
$\tau=\tilde{\tau}$ where 
$\tilde{\tau} = \inf \{t >0 : \tilde{X}_t \leq K\}$. 
Hence we have 
\begin{equation*}
\mathbf{E}[f(\tilde{X}_T)I_{\{\tilde{\tau} >T\}}] 
= \mathbf{E}[f(X_T)I_{\{\tau >T\}}]. 
\end{equation*}
\end{proof}
\subsection{Important Remark}
We do not anymore 
regard (\ref{new}) as an equation for semi-static hedging
but a relation to give a \underline{numerical scheme} to calculate the expectation of 
the diffusion with a Dirichlet boundary condition in terms of 
those without boundary conditions. 
The former is very difficult while the latter is rather easier 
using rapidly developing technique from numerical finance. 
In the following sections, we will present some 
results of numerical examples to show the 
effectivity of the new scheme. 

\section{Numerical Experiments for One Dimensional Models}\label{sec-numerical}
\subsection{The Euler-Maruyama Scheme}\label{sec-EM}
Here we briefly recall the Euler-Maruyama scheme for 
a diffusion process given as a solution to SDE $(\ref{diffusion})$. 
Fix $T>0$. 
For $n \geq 1$, we set a subdivision of the 
interval $[0,T]$ 
\begin{equation*}
0 = t_0 \leq t_1 \leq \cdots \leq t_{n} =T, 
\end{equation*}
where $t_k := \frac{kT}{n}$ for $0 \leq k \leq n$, 
and we denote this net by $\bigtriangleup_n$. 

The Euler-Maruyama scheme is a general method for numerically solving 
(\ref{diffusion}) 
by a discretized stochastic process which is given by
\begin{equation}\label{ES}
X^n_{t_{k+1}} 
= X^n_{t_k} + \sigma(X^n_{t_k})(t_{k+1}-t_k) +\mu(X^n_{t_k})
(W_{t_{k+1}} - W_{t_k}), 
\end{equation}
$ k=0,1,2,\cdots,n-1 $, 
and for $t_k < t < t_{k+1} $, $ X_t $ is given by
an interpolation. 
The approximating process $(X^n_T)$ is simulated by using 
independent quasi-random Gaussian variables for the increments 
$(W_{t_{k+1}} - W_{t_k})_{0 \leq k \leq n-1}$. 

We rely on the following result; 
\begin{thm}[Theorem $3.1$, \citet{Y}]\label{thm-Y}
If the set of discontinuous points of $\sigma$ and $\mu$ is 
countable, then the Euler scheme (\ref{ES}) converges weakly to the 
unique weak solution of SDE (\ref{diffusion}) as $n \rightarrow \infty$. 
\end{thm}

From now on, in addition to {\bf (H1)} and {\bf (H2)}, we assume that 
$\sigma$ and $\mu$ are piece-wise continuous. 
\subsubsection{Path-wise Method}
Since the convergence is 
in the space of probability measures on continuous functions, 
we see that this algorithm 
can also be used to simulate a path-dependent 
functional of the process; 
in particular, 
$f(X_T) I_{\{\tau >T\}}$, where $ f $ is a (bounded) 
continuous function and $\tau$ is the first hitting time 
defined by $(\ref{tau})$. 
The functional is approximated by 
$f(X_T^n)I_{\{\tau^n > T\}}$, where 
$\tau^n:= \inf\{t_k: X_{t_k}^n \leq K\}$ is 
the discretized first hitting time to $K$. 
Then the expectation  
$\bE[f(X_T)I_{\{\tau >T \}}]$ 
is approximated with a Monte-Carlo algorithm by 
\begin{met}
(Path-wise EM scheme)\\
\begin{equation}\label{Path-wise}
\frac{1}{M} \sum_{i=1}^M 
f(X_T^n(\omega_i)) 
I_{\{\tau^n (\omega_i)>T\}}.
\end{equation}
\end{met}
By the strong law of large numbers, (\ref{Path-wise}) 
converges to $\bE[f(X_T^n)I_{\{\tau^n >T \}}]$ 
as $M$ goes to infinity.  
Moreover, as the index $n$ of the net $\bigtriangleup_n$ goes to infinity, 
$\bE[f(X_t^n)I_{\{\tau^n >T \}}]$ converges to 
$\bE[f(X_T) I_{\{\tau >T\}}]$. s
According to \citet{G}, the following convergence rate was given;
\begin{thm}[Theorem $2.3$, \citet{G}]
Assume that 
$\sigma$ and $\mu$ are in $C_b^{\infty}$, $\sigma$ is bounded below from zero
 and a solution is non-explosion. 
Then for a bounded measurable function $f$ such that $ d ({\rm supp}f, K) >0$, 
there is a constant $C$ such that 
\begin{equation*}
|\bE[f(X_t^n)I_{\{\tau^n >T \}}] - \bE[f(X_T) I_{\{\tau >T\}}]| 
< C \frac{1}{\sqrt{n}}.
\end{equation*}
\end{thm}


\subsubsection{Put-Call Symmetry Method}\label{sec-PCSM}
Let $\tilde{X}$ be a solution with coefficients $\tilde{\sigma}$ and 
$\tilde{\mu}$ given by (\ref{sigma2}) and (\ref{mu2}), and 
$(\tilde{X}_t^n)$ be the discretized Euler-Maruyama process 
with respect to the net $\bigtriangleup_n$.
Namely, 
\begin{equation*}
\begin{split}
\tilde{X}^n_{t_{k+1}} 
&=\tilde{X}^n_{t_k} +( \sigma(\tilde{X}^n_{t_k})(t_{k+1}-t_k) 
+\mu(\tilde{X}^n_{t_k})(W_{t_{k+1}} - W_{t_k}))I_{\{\tilde{X}^n_{t_k}> K\}}\\
& \quad
+(\sigma(2K-\tilde{X}^n_{t_k})(t-t_k) 
-\mu(2K-\tilde{X}^n_{t_k})(W_t - W_{t_k}))I_{\{\tilde{X}^n_{t_k} \leq K\}}
\end{split}\end{equation*}
for $k=0,1,2,\cdots,n-1$. 
With an interpolation, $\tilde{X}_t^n$ for  
$t_k \leq t \leq t_{k+1}$ is obtained as well. 
Since the set of the discontinuous in the coefficients is 
Lebesgue measure zero, $\tilde{X}^n$ also converges weakly 
to $\tilde{X}$ by Theorem \ref{thm-Y}. 

Combining Theorem \ref{Thm} and \ref{thm-Y}, 
we may rely on the following algorithm; 
the expectation $\bE[f(X_T)I_{\{\tau >T \}}]$ 
is approximated with Monte-Carlo algorithm by 
\begin{met}
(Put-Call symmetry method)\\
\begin{equation}\label{ES-symmetrize}
\frac{1}{M}\sum_{i=1}^M 
\left\{
f(\tilde{X}_T^n(\omega_i))I_{\{\tilde{X}_T^n(\omega_i)>K\}}
-f(2K-\tilde{X}_T^n(\omega_i))I_{\{\tilde{X}_T^n(\omega_i)<K\}}
\right\}.
\end{equation}
\end{met}
As $M$ goes to $\infty$, (\ref{ES-symmetrize}) converges to 
\begin{equation}\label{ES-2}
\mathbf{E}[f(\tilde{X}^n_T)I_{\{\tilde{X}^n_T>K\}}]
-\mathbf{E}[f(2K-\tilde{X}^n_T)I_{\{\tilde{X}^n_T \leq K\}}]. 
\end{equation}
By the weak convergence of $X^n$, 
$(\ref{ES-2})$ converges to 
\begin{equation*}
\mathbf{E}[f(\tilde{X}_T)I_{\{\tilde{X}_T>K\}}]
-\mathbf{E}[f(2K-\tilde{X}_T)I_{\{\tilde{X}_T \leq K\}}], 
\end{equation*}
as $n \rightarrow \infty$. 
However, we don't know the exact rate of convergence in this algorithm 
since the coefficients are inevitably non-smooth at $K$\footnote{ 
There are many results on the rate of convergence 
when $\sigma$ and $\mu$ are smooth. For example, 
when $\sigma$ and $\mu$ are in $\cC_p^4$, the space of 
functions such that $4$-th derivative exists and have a polynomial growth, 
we have 
\begin{equation*}
|\bE[g(X^n_T)] - \bE[g(X_T)]| = O(n),
\end{equation*}
for any $g \in \cC_p^4$ (See Kloeden and Platen \citet{KP}, pp. 476). }.

The numerical results in the next section, however, may imply 
that 
the convergence rate of Put-Call symmetry method 
is better than that of the path-wise EM scheme. 
To prove this {\em conjecture} would be a very interesting 
mathematical challenge. 

\subsection{Numerical Results}
In this section, we give numerical examples 
using method $1$ (path-wise EM method) 
and method $2$ (Put-Call symmetry method) 
under Black-Scholes model and other CEV models. 
Let us consider the value of 
barrier call option with strike price $S$ 
and knockout barrier $K$. 
\subsubsection{Black-Scholes Model}\label{BS-result}
The underlying price process of Black-Scholes model 
is the unique solution of the following SDE; 
\begin{equation}
dX_t = rX_t dt + \sigma X_t dW_t,
\end{equation} 
for $r,\ \sigma \geq 0$. Then the value of barrier option 
is accurately-calculable since the joint distribution 
of Brownian motion and the hitting time of Brownian motion to a point is 
computable by using the reflection principle. The exact option price is 
given by the following;
\begin{equation*}
\mathbf{E}[(X_T - S)^+ I_{\{\tau >T\}}]
=V_{call}(X_0)- \left(\frac{K}{X_0}\right)^{\frac{2r}{\sigma^2}-1} 
V_{call}\left(\frac{K^2}{X_0}\right),\\
\end{equation*}
where 
\begin{eqnarray*}
\begin{split}
V_{call}(x)&= 
x e^{rT}\left(1- \Phi(d_{+}(x))\right)- S\left( 1- \Phi(d_{-}(x))\right),\\
d_{\pm}(x)&=  \frac{\log(\frac{S}{x}) 
-\left(r\pm \frac{\sigma^2}{2}\right)T}{\sigma\sqrt{T}},
\end{split}
\end{eqnarray*}
and $\Phi$ is the distribution 
function of the standard normal distribution. 

Fix a maturity time $T>0 $. 
Tables \ref{table1}-\ref{table4} give simulation results 
for the value of down-and-out call option 
$\mathbf{E}[(X_T - S)^+ I_{\{\tau>T\}}]$ 
under the path-wise Euler-Maruyama method (EM) 
and the Put-Call symmetry method (PCM). 
We take $[X_0=100,\ S=95,\ K=90,\ T=1],$ and 
\begin{enumerate}
\item[Table] \ref{table1}: $\sigma = 0.2,\ r = 0$, 
\item[Table] \ref{table2}: $\sigma = 0.2,\ r = 0.02$, 
\item[Table] \ref{table3}: $\sigma = 0.5,\ r = 0$, 
\item[Table] \ref{table4}: $\sigma = 0.5,\ r = 0.02$. 
\end{enumerate}
In the PCM, we symmetrize the functions 
$\mu(x) = rx$ and $\sigma(x)= \sigma x $ at $ K $. 

The number of simulation trials is set equal to 
the cube of the number 
of time steps for the Euler discretization. 
The errors in the last two columns 
are calculated as
$$ \frac{| \text{EM} (\text{PCM}) - 
\mbox{true option price}|}{\mbox{ true option price}} . $$

\begin{table}[htbp]
\caption{
Black-Scholes model
\protect\\ \footnotesize{
$X_0=100,\ S=95,\ K=90,\ \sigma = 0.2,\ r = 0,\ T=1,$ 
option price = $8.17140$}}
\begin{center}
\begin{tabular}{c c|c c|c c}
\hline
 No. of & No. of &  &&EM & PCM \\ 
simulation trials & time steps &EM & PCM & 
error(\%) & error(\%) \\ \hline
$1000$ &$10 $ & $8.881 $ & $7.816 $ & $  8.7 $ & $  4.4 $ \\
 $8000$ &$20 $ & $9.183 $ & $8.172 $ & $ 12.4 $ & $  0.0 $ \\
 $27000$ &$30 $ & $8.992 $ & $8.250 $ & $ 10.0 $ & $  1.0 $ \\
 $64000$ &$40 $ & $8.880 $ & $8.175 $ & $  8.7 $ & $  0.0 $ \\
 $125000$ &$50 $ & $8.804 $ & $8.190 $ & $  7.7 $ & $  0.2 $ \\
 $216000$ &$60 $ & $8.692 $ & $8.137 $ & $  6.4 $ & $  0.4 $ \\
 $343000$ &$70 $ & $8.697 $ & $8.127 $ & $  6.4 $ & $  0.5 $ \\
 $512000$ &$80 $ & $8.671 $ & $8.171 $ & $  6.1 $ & $  0.0 $ \\
 $729000$ &$90 $ & $8.672 $ & $8.207 $ & $  6.1 $ & $  0.4 $ \\
 $1000000$ &$100 $ & $8.597 $ & $8.135 $ & $  5.2 $ & $  0.4 $ \\
\hline
\end{tabular}
\end{center}
\label{table1}
\caption{
Black-Scholes model
\protect\\ 
\footnotesize{
$X_0=100,\ S=95,\ K=90,\ \sigma = 0.2,\ r = 0.02,\ T=1,$ 
option price = $9.31138$}}
\begin{center}
\begin{tabular}{c c|c c|c c}
\hline
 No. of & No. of &  &&EM & PCM \\ 
simulation trials & time steps &EM & PCM & 
error(\%) & error(\%) \\ \hline
$1000$ &$10 $ & $10.953 $ & $9.821$ & $ 17.6 $ & $  5.5 $ \\
 $8000$ &$20 $ & $10.050$ & $9.165$ & $ 7.9 $ & $  1.6$ \\
 $27000$ &$30 $ & $10.090$ & $9.226 $ & $  8.4 $ & $  0.9 $ \\
 $64000$ &$40 $ & $9.952 $ & $9.258 $ & $  6.9 $ & $  0.6 $ \\
 $125000$ &$50 $ & $9.974 $ & $9.302 $ & $  7.1 $ & $  0.1 $ \\
 $216000$ &$60 $ & $10.033 $ & $9.389 $ & $  7.7 $ & $  0.8 $ \\
 $343000$ &$70 $ & $9.911 $ & $9.298 $ & $  6.4 $ & $  0.1 $ \\
 $512000$ &$80 $ & $9.885 $ & $9.353 $ & $  6.2 $ & $  0.4 $ \\
 $729000$ &$90 $ & $9.839 $ & $9.306 $ & $  5.7 $ & $  0.1 $ \\
 $1000000$ &$100 $ & $9.811 $ & $9.309 $ & $  5.4 $ & $  0.0 $ \\
\hline
\end{tabular}
\end{center}
\label{table2}
\end{table}
\begin{table}[htbp]
\caption{
Black-Scholes model
\protect\\ 
\footnotesize{
$X_0=100,\ S=95,\ K=90,\ \sigma = 0.5,\ r = 0,\ T=1,$ 
option price = $9.37170$}}
\begin{center}
\begin{tabular}{c c|c c|c c}
\hline
 No. of & No. of &  &&EM & PCM \\ 
simulation trials & time steps &EM & PCM & 
error(\%) & error(\%) \\ \hline
$1000$ &$10 $ & $15.981 $ & $9.521 $ & $ 70.5 $ & $  1.6 $ \\
 $8000$ &$20 $ & $14.455 $ & $9.742 $ & $ 54.2 $ & $  4.0 $ \\
 $27000$ &$30 $ & $13.074 $ & $9.126 $ & $ 39.5 $ & $  2.6 $ \\
 $64000$ &$40 $ & $12.837 $ & $9.479 $ & $ 37.0 $ & $  1.1 $ \\
 $125000$ &$50 $ & $12.281 $ & $9.251 $ & $ 31.0 $ & $  1.3 $ \\
 $216000$ &$60 $ & $11.942 $ & $9.231 $ & $ 27.4 $ & $  1.5 $ \\
 $343000$ &$70 $ & $11.838 $ & $9.307 $ & $ 26.3 $ & $  0.7 $ \\
 $512000$ &$80 $ & $11.750 $ & $9.450 $ & $ 25.4 $ & $  0.8 $ \\
 $729000$ &$90 $ & $11.549 $ & $9.392 $ & $ 23.2 $ & $  0.2 $ \\
 $1000000$ &$100 $ & $11.443 $ & $9.319 $ & $ 22.1 $ & $  0.6 $ \\
\hline
\end{tabular}
\end{center}
\label{table3}
\caption{
Black-Scholes model
\protect\\
\footnotesize{
$X_0=100,\ S=95,\ K=90,\ \sigma = 0.5,\ r = 0.02,\ T=1,$ 
option price = $10.02470$}}
\begin{center}
\begin{tabular}{c c|c c|c c}
\hline
 No. of & No. of &  &&EM & PCM \\ 
simulation trials & time steps &EM & PCM & 
error(\%) & error(\%) \\ \hline
$1000$ &$10 $ & $15.488 $ & $9.688 $ & $ 54.5 $ & $  3.4 $ \\
 $8000$ &$20 $ & $14.687 $ & $9.540 $ & $ 46.5 $ & $  4.8 $ \\
 $27000$ &$30 $ & $14.065 $ & $10.341 $ & $ 40.3 $ & $  3.2 $ \\
 $64000$ &$40 $ & $13.472 $ & $10.191 $ & $ 34.4 $ & $  1.7 $ \\
 $125000$ &$50 $ & $13.012 $ & $9.779 $ & $ 29.8 $ & $  2.4 $ \\
 $216000$ &$60 $ & $12.981 $ & $10.257 $ & $ 29.5 $ & $  2.3 $ \\
 $343000$ &$70 $ & $12.707 $ & $9.991 $ & $ 26.8 $ & $  0.3 $ \\
 $512000$ &$80 $ & $12.391 $ & $9.916 $ & $ 23.6 $ & $  1.1 $ \\
 $729000$ &$90 $ & $12.418 $ & $10.098 $ & $ 23.9 $ & $  0.7 $ \\
 $1000000$ &$100 $ & $12.235 $ & $10.025 $ & $ 22.0 $ & $  0.0 $ \\\hline
\end{tabular}
\end{center}
\label{table4}
\end{table}

One sees that, in the experiments, the Put-Call symmetry method always beats 
the path-wise EM method. 

\subsubsection{CEV Model}\label{cevsim}
Here the underlying price process
is a solution of the following SDE; 
\begin{equation}\label{CEV}
dX_t = rX_t dt + \sigma X_t^\beta dW_t,
\end{equation} 
for $r,\ \sigma \geq 0$ and $ \beta \geq \frac{1}{2}$. 
We take $ 0.75 $ in the experiments. 

Tables \ref{table7}-\ref{table8} are simulation results 
for down-and-out call options with EM and PCM. 
We set parameters to 
$[X_0=100,\ S=95,\ K=90,\ \beta= 0.75,\ \sigma = 0.45,\ T=1],$ and 
\begin{enumerate}
\item[Table] \ref{table7}: $r=0$, 
\item[Table] \ref{table8}: $r=0.02$, 
\end{enumerate} 
in the experiments. 
For CEV model, we do not have an analytic formula. So, 
as a benchmark, we used numerical results by 
the path-wise Euler-Maruyama scheme 
where the number of 
time steps for the Euler discretization is $5,000$ 
and that of a Monte-Carlo simulation is $50,000,000$. 
Note that since we are calculating the prices of
down-and-out call options, we do not need to care about 
the singularity at 
$ x = 0 $ in the SDE. 

\begin{table}[htbp]
\caption{
CEV model
\protect\\ \footnotesize{
$X_0=100,\ S=95,\ K=90,\ \beta= 0.75,\ \sigma = 0.45,\ r = 0,\ T=1,$ 
\protect\\ benchmark of option price = $7.50095$}}
\begin{center}
\begin{tabular}{c c|c c|c c}
\hline
 No. of & No. of &  &&EM & PCM \\ 
simulation trials & time steps &EM & PCM & 
error(\%) & error(\%) \\ \hline
$1000$ &$10 $ & $7.781 $ & $7.068 $ & $  3.8 $ & $  5.7 $ \\
 $8000$ &$20 $ & $7.997 $ & $7.504 $ & $  6.6 $ & $  0.1 $ \\
 $27000$ &$30 $ & $7.805 $ & $7.397 $ & $  4.1 $ & $  1.4 $ \\
 $64000$ &$40 $ & $7.758 $ & $7.379 $ & $  3.5 $ & $  1.6 $ \\
 $125000$ &$50 $ & $7.730 $ & $7.412 $ & $  3.1 $ & $  1.2 $ \\
 $216000$ &$60 $ & $7.733 $ & $7.407 $ & $  3.1 $ & $  1.2 $ \\
 $343000$ &$70 $ & $7.714 $ & $7.422 $ & $  2.9 $ & $  1.0 $ \\
 $512000$ &$80 $ & $7.691 $ & $7.423 $ & $  2.6 $ & $  1.0 $ \\
 $729000$ &$90 $ & $7.680 $ & $7.414 $ & $  2.4 $ & $  1.1 $ \\
 $1000000$ &$100 $ & $7.654 $ & $7.414 $ & $  2.1 $ & $  1.1 $ \\
\hline
\end{tabular}
\end{center}
\label{table7}
\caption{
CEV model
\protect\\ 
\footnotesize{
$X_0=100,\ S=95,\ K=90,\ \beta= 0.75,\ \sigma = 0.45,\ r = 0.02,\ T=1,$ 
\protect\\ benchmark of option price = $8.82718$}}
\begin{center}
\begin{tabular}{c c|c c|c c}
\hline
 No. of & No. of &  &&EM & PCM \\ 
simulation trials & time steps &EM & PCM & 
error(\%) & error(\%) \\ \hline
$1000$ &$10 $ & $9.418 $ & $8.918 $ & $  6.7 $ & $  1.0 $ \\
 $8000$ &$20 $ & $9.349 $ & $8.986 $ & $  5.9 $ & $  1.8 $ \\
 $27000$ &$30 $ & $9.242 $ & $8.791 $ & $  4.7 $ & $  0.4 $ \\
 $64000$ &$40 $ & $9.193 $ & $8.772 $ & $  4.1 $ & $  0.6 $ \\
 $125000$ &$50 $ & $9.109 $ & $8.760 $ & $  3.2 $ & $  0.8 $ \\
 $216000$ &$60 $ & $9.089 $ & $8.751 $ & $  3.0 $ & $  0.9 $ \\
 $343000$ &$70 $ & $9.063 $ & $8.742 $ & $  2.7 $ & $  1.0 $ \\
 $512000$ &$80 $ & $9.009 $ & $8.722 $ & $  2.1 $ & $  1.2 $ \\
 $729000$ &$90 $ & $9.027 $ & $8.745 $ & $  2.3 $ & $  0.9 $ \\
 $1000000$ &$100 $ & $8.995 $ & $8.722 $ & $  1.9 $ & $  1.2 $ \\
\hline
\end{tabular}
\end{center}
\label{table8}
\end{table}


\section{Put-Call Symmetry Method Applied to Stochastic Volatility Models}\label{sec-svm}
In this section, we slightly extend the put-call symmetry method 
to apply it to 
stochastic volatility models which are described by two-dimensional 
SDE. Theoretical backgrounds of the extension is given in \citet{AI}. 

A generic stochastic volatility model is given as follows;
\begin{equation}\label{GSV1}
\begin{split}
dX_t &= \sigma_{11} (X_t,V_t) dW_t  + \mu_1 (X_t,V_t) \,dt  \\
dV_t &= \sigma_{21} (V_t) dW_t + \sigma_{22} (V_t) dB_t + \mu_2 (V_t) \, dt,
\end{split}
\end{equation}
where $ W $ and $ B $ are mutually independent ($1$-dim) Wiener processes, 
\begin{equation*}
\sigma (x, v) = 
\begin{pmatrix}
\sigma_{11} (x,v) & 0 \\
\sigma_{21} (v) & \sigma_{22} (v)  \\
\end{pmatrix}
\end{equation*}
and $ \mu (x,v) = (\mu_1 (x,v), \mu_2 (v) ) $
are continuous functions 
on $ \mathbf{R}^2 $.  
Here we simply assume that $\sigma$ and 
$\mu$ are sufficiently regular (not so irregular) 
to allow a unique weak solution in (\ref{GSV1}). 
The independence of $ V $ against $ S $ plays an important role
in applying our scheme. In fact, thanks to the property, we 
may simply work on the symmetrization with respect to the reflection 
$(x,y) \mapsto (2K-x,y)$.
Let us be more precise. 
Let $ (X, V) $ be a $2$-dimensional diffusion process given 
as a (weak) unique solution of SDE (\ref{GSV1}), and  
$ \tau $ be
the first passage time of $X$ to $K$. 
We note that $\tau$ is not dependent on $V$.
We say that {\em arithmetic put-call symmetry} at $ K $  
holds for $ (X,V) $ if 
\begin{equation*}
(X_t,V_t) 1_{ \{ \tau \leq t \} }
\overset{\mathop{d}}{=} (2K -X_t,V_t) 1_{ \{\tau \leq t \} } 
\end{equation*}
for any $ t > 0 $. 

Mathematically, we rely on the following result from \citet{AI}. 
\begin{prop}[\citet{AI}]\label{prop-CLSV}
If the coefficients satisfy the following conditions;  
\begin{equation}\label{SVsigma}
\sigma_{11} (x,v) = -\sigma_{11} (2K-x,v),
\end{equation}
\begin{equation}\label{SVmu}
\mu_1(x,v)=-\mu(2K-x,v),
\end{equation}
for $(x,v) \in (\bR\setminus\{ K \})\times \bR$, 
then APCS at $ K $ holds for $ (X,V) $.
\end{prop}
On the basis of Proposition \ref{prop-CLSV}, we 
construct another diffusion $ (\tilde{X},V) $ that satisfies
APCS at $ K $ in a totally similar way as the one dimensional case, and 
we obtain a static hedging formula corresponding to Theorem \ref{Thm}. 
\begin{prop}\label{SVPCS}
Let $ K > 0 $ and put 
\begin{equation*}
\tilde{\sigma}_{11} (x,v) = 
\begin{cases}
\sigma_{11} (x,v) & x \geq K \\
- \sigma_{11} (2K-x,v) & x < K
\end{cases},
\end{equation*}
\begin{equation*}
\tilde{\mu}_{1} (x,v) = 
\begin{cases}
\mu_{1} (x,v) & x \geq K \\
- \mu_{1} (2K-x,v) & x < K 
\end{cases},
\end{equation*}
and let $ \tilde{X} $ be the unique (weak) solution to 
\begin{equation*}
d\tilde{X}_t = \tilde{\sigma}_{11} (\tilde{X}_t,V_t) dW_t  
+ \tilde{\mu}_1 (\tilde{X}_t,V_t) \,dt,
\end{equation*}
where $ V $ is the solution to (\ref{GSV1}). 
Then, it holds for any bounded Borel function $ f $ and $ t>0 $ that
\begin{equation}\label{STH7}
\begin{split}
& E [ f( X_t )1_{\{X_t>K\}} 1_{\{ \tau_K > t\}}]  \\
& \qquad = 
E[f(\tilde{X}_t)1_{\{\tilde{X}_t > K\}} ] 
- E [f (2K-\tilde{X}_t)1_{\{ \tilde{X}_t < K \}}],
\end{split}
\end{equation}
where $ X $ is the solution to (\ref{GSV1}) with $ X_0 > K $. 
\end{prop}
\begin{proof}
Omitted. 
\end{proof}

\subsection{Numerical Results on Single Barrier Options 
under Stochastic Volatility Models}\label{num;sv;single}
In this section we 
give numerical examples 
of the price of a single barrier option under 
Heston's and SABR type stochastic volatility models,
using numerical method based on (\ref{STH7}). 

The Euler-Maruyama scheme of 
the solution of SDE (\ref{GSV1})
with respect to the net $\bigtriangleup_n=\{t_0,t_1,\cdots, t_n\}$  
is given by the following; 
\begin{equation*}
\begin{split}
X^n_{t_{k+1}} 
&= X^n_{t_k} + \sigma_{11}(X^n_{t_k}, V^n_{t_k})(W_{t_{k+1}} - W_{t_k}) 
+\mu_1(X^n_{t_k},V^n_{t_k})(t_{k+1}-t_k),\\
V^n_{t_{k+1}}
&= V^n_{t_k} + \sigma_{21}(V^n_{t_k})(W_{t_{k+1}} - W_{t_k}) 
+\sigma_{22}(V^n_{t_k})(B_{t_{k+1}} - B_{t_k})\\
&\quad  
+\mu_2(V^n_{t_k})(t_{k+1}-t_k),
\end{split}
\end{equation*}
for $ k=0,1,2,\cdots,n-1$. 
With an interpolation, $ X_t $ for $t_k < t \leq t_{k+1} $
is obtained as well.   
Here $W$ and $B$ denotes 
two independent $1$-dimensional Brownian motions. 
The increments $W_{t_{k+1}} - W_{t_k}$ and $B_{t_{k+1}} - B_{t_k}$ 
are simulated by quasi-random independent Gaussian variables.

The underlying price process of Heston model 
is given as follows;
\begin{equation}
\begin{cases}
dX_t^{}=rX_t dt+\sqrt{V_t^{}}X_t^{}dW_{t},\\
dV _t^{}=\kappa(\theta-V_t^{})dt +\nu \sqrt{V _t}^{}(\rho dW_t+\sqrt{{1-\rho^2}}dB_t) \label{V}
\end{cases} 
\end{equation}
for $r,\, \kappa,\, \theta,\, \nu >0 $ and $-1 \leq \rho \leq 1$. 
Then the symmetrized path $\tilde{X}$ 
is constructed as a solution to the following SDE;
\begin{eqnarray*}
\begin{split}
d\tilde{X_t}^{}&= \left( r \tilde{X_t}I_{\{ \tilde{X_t}> K\}} 
- r (2K-\tilde{X_t}) I_{\{ \tilde{X_t}<K\}}  \right) dt \\ 
& \quad 
+ \left( {\sqrt{V_t}}\tilde{X_t}I_{\{ \tilde{X_t}> K\}}  
- \sqrt{V_t} (2K-\tilde{X_t})I_{\{ \tilde{X_t}< K\}} \right) dW_t, 
\end{split}
\end{eqnarray*}
where $V$ is the solution of SDE (\ref{V}). 

The underlying asset price of $\lambda$-SABR model is described as 
\begin{equation}
\begin{cases}
dX_t^{}=rX_t dt+V_t^{}X_t^{\beta}dW_t,\\
dV_t^{}=\lambda(\theta-V_t) dt + \nu V _t^{}(\rho dW_t+\sqrt{1-\rho ^2}dB_t)
\label{V1}
\end{cases} 
\end{equation}
for $r,\, \lambda,\, \theta,\, \nu >0 $, $\beta \geq \frac{1}{2}$ and $-1 \leq \rho \leq 1$. 
Then the symmetrized process $\tilde{X}$ is given by the following SDE; 
\begin{eqnarray*}
\begin{split}
d\tilde{X_t}^{}&=\left( r\tilde{X_t}I_{\{ \tilde{X_t}>K\}}
- r(2K-\tilde{X_t}) I_{\{ \tilde{X_t}< K\}}\right) dt \\ 
& \quad 
+\left( V_t^{} {\tilde{X_t}}^\beta I_{\{\tilde{X_t}> K\}}
- V_t (2K-\tilde{X_t})^\beta I_{\{\tilde{X_t} < K\}} \right) dW_{1,t},
\end{split}
\end{eqnarray*}
where $V$ is the solution of SDE (\ref{V1}). 

Tables \ref{sim-h} - \ref{sim-r} below
are simulation results of 
the price of a single barrier call option 
under Heston's and $\lambda$-SABR model, respectively.
We set the parameters as 
$[X_0=100,\ V_0=0.03,\ K=95,\ H=90,\ \theta = 0.03,\ r = 0,\ 
T=1,\ \kappa =1,\ 
\rho=-0.7,\ \nu = 0.03]$ in Heston model 
(Table \ref{sim-h} and Table \ref{sim-h2}), 
and   
$[X_0=100,\ V_0= 0.5,\ 
S=95,\ K=90,\ \theta = 0.03,\ 
r = 0,\ T=1,\ \beta =0.75,\ 
\lambda = 1.0,\ 
\rho=-0.7,\ \nu = 0.3] $ 
in $\lambda$-SABR model (Table \ref{sim-r1} and Table \ref{sim-r1}), and 
\begin{enumerate}
\item[Table] \ref{sim-h} and \ref{sim-r1}: $r=0$, 
\item[Table] \ref{sim-h2} and \ref{sim-r}: $r=0.02$, 
\end{enumerate} 
in the experiments. 
Benchmark is given in the same setting of Section \ref{cevsim}. 
\begin{table}[htbp]
\caption{Heston model
\protect\\ 
\footnotesize{ 
$X_0=100,\ V_0=0.03,\ K=95,\ H=90,\ \theta = 0.03,\ r = 0,\ 
T=1,\ \kappa =1,\ 
\rho=-0.7,\ \nu = 0.03,\  
\mbox{benchmark\ of\ option \ price } = 7.92706$}}
\begin{center}
\begin{tabular}{c c|c c|c c}
\hline
 No. of & No. of &  &&EM & PCM \\ 
simulation trials & time steps &EM & PCM & 
error(\%) & error(\%) \\ \hline
$1000$ &$10 $ & $8.638 $ & $7.953 $ & $9.0 $ & $0.3$ \\
$8000$ &$20 $ & $8.761 $ & $8.167 $ & $10.5 $ & $3.0$ \\
$27000$ &$30 $ & $8.466 $ & $7.932 $ & $6.8 $ & $0.1$ \\
$64000$ &$40 $ & $8.477 $ & $8.017 $ & $6.9 $ & $1.1$ \\
$125000$ &$50 $ & $8.366 $ & $7.892 $ & $5.5 $ & $0.4$ \\
$216000$ &$60 $ & $8.301 $ & $7.877 $ & $4.7 $ & $0.6$ \\
$343000$ &$70 $ & $8.246 $ & $7.875 $ & $4.0 $ & $0.7$ \\
$512000$ &$80 $ & $8.273 $ & $7.902 $ & $4.4 $ & $0.3$ \\
$729000$ &$90 $ & $8.221 $ & $7.875 $ & $3.7 $ & $0.7$ \\
$1000000$ &$100 $ & $8.212 $ & $7.871 $ & $3.6 $ & $0.7$ \\
\hline
\end{tabular}
\end{center}
\label{sim-h}
\caption{Heston model
\protect\\ 
\footnotesize{ 
$X_0=100,\ V_0=0.03,\ K=95,\ H=90,\ \theta = 0.03,\ r = 0.02,\ 
T=1,\ \kappa =1,\ 
\rho=-0.7,\ \nu = 0.03,\  
\mbox{benchmark\ of\ option \ price } = 9.15602$}}
\begin{center}
\begin{tabular}{c c|c c|c c}
\hline
 No. of & No. of &  &&EM & PCM \\ 
simulation trials & time steps &EM & PCM & 
error(\%) & error(\%) \\ \hline
$1000$ &$10 $ & $10.308 $ & $9.192 $ & $12.6 $ & $0.4$ \\
$8000$ &$20 $ & $9.828 $ & $9.197 $ & $7.3 $ & $0.5$ \\
$27000$ &$30 $ & $9.572 $ & $8.953 $ & $4.5 $ & $2.2$ \\
$64000$ &$40 $ & $9.674 $ & $9.133 $ & $5.7 $ & $0.3$ \\
$125000$ &$50 $ & $9.632 $ & $9.134 $ & $5.2 $ & $0.2$ \\
$216000$ &$60 $ & $9.552 $ & $9.093 $ & $4.3 $ & $0.7$ \\
$343000$ &$70 $ & $9.525 $ & $9.096 $ & $4.0 $ & $0.7$ \\
$512000$ &$80 $ & $9.524 $ & $9.135 $ & $4.0 $ & $0.2$ \\
$729000$ &$90 $ & $9.498 $ & $9.116 $ & $3.7 $ & $0.4$ \\
$1000000$ &$100 $ & $9.454 $ & $9.106 $ & $3.2 $ & $0.5$ \\
\hline
\end{tabular}
\end{center}
\label{sim-h2}
\end{table}
\begin{table}[htbp]
\caption{
$\lambda$-SABR model 
\protect\\ 
\footnotesize{
$X_0=100,\ V_0= 0.5,\ 
S=95,\ K=90,\ \theta = 0.03,\ 
r = 0,\ T=1,\ \beta =0.75,\ 
\lambda = 1.0,\ 
\rho=-0.7,\ \nu = 0.3,\ 
\mbox{benchmark\ of\ option\ price } = 6.59534$}}
\begin{center}
\begin{tabular}{c c|c c|c c}
\hline
 No. of & No. of &  &&EM & PCM \\ 
simulation trials & time steps &EM & PCM & 
error(\%) & error(\%) \\ \hline
$1000$ &$10$  & $6.643 $ & $6.478 $ & $0.7 $ & $1.8$ \\
$8000$ & $20$ & $6.708 $ & $6.591 $ & $1.7 $ & $0.1$ \\
$27000$& $30$  & $6.701 $ & $6.584 $ & $1.6 $ & $0.2$ \\
$64000 $ & $40$& $6.671 $ & $6.565 $ & $1.1 $ & $0.5$ \\
$125000 $ & $50$& $6.668 $ & $6.568 $ & $1.1 $ & $0.4$ \\
$216000 $  &$60$& $6.672 $ & $6.581 $ & $1.2 $ & $0.2$ \\
$343000 $ & $70$& $6.669 $ & $6.585 $ & $1.1 $ & $0.2$ \\
$512000 $  &$80$& $6.671 $ & $6.597 $ & $1.1 $ & $0.0$ \\
$729000 $ &$90$ & $6.655 $ & $6.579 $ & $0.9 $ & $0.2$ \\
$1000000 $& $100$ & $6.646 $ & $6.576 $ & $0.8 $ & $0.3$ \\
\hline
\end{tabular}
\end{center}
\label{sim-r1}
\caption{
$\lambda$-SABR model 
\protect\\ 
\footnotesize{
$X_0=100,\ V_0= 0.5,\ 
S=95,\ K=90,\ \theta = 0.03,\ 
r = 0.02,\ T=1,\ \beta =0.75,\ 
\lambda = 1.0,\ 
\rho=-0.7,\ \nu = 0.3,\ 
\mbox{benchmark\ of\ option\ price } = 8.71005$}}
\begin{center}
\begin{tabular}{c c|c c|c c}
\hline
 No. of & No. of &  &&EM & PCM \\ 
simulation trials & time steps &EM & PCM & 
error(\%) & error(\%) \\ \hline
$1000$ &$10 $ & $9.493 $ & $8.779 $ & $9.0 $ & $0.8$ \\
$8000$ &$20 $ & $9.081 $ & $8.582 $ & $4.3 $ & $1.5$ \\
$27000$ &$30 $ & $9.106 $ & $8.723 $ & $4.5 $ & $0.2$ \\
$64000$ &$40 $ & $9.029 $ & $8.656 $ & $3.7 $ & $0.6$ \\
$125000$ &$50 $ & $9.007 $ & $8.683 $ & $3.4 $ & $0.3$ \\
$216000$ &$60 $ & $9.008 $ & $8.710 $ & $3.4 $ & $0.0$ \\
$343000$ &$70 $ & $8.988 $ & $8.707 $ & $3.2 $ & $0.0$ \\
$512000$ &$80 $ & $8.940 $ & $8.670 $ & $2.6 $ & $0.5$ \\
$729000$ &$90 $ & $8.923 $ & $8.671 $ & $2.4 $ & $0.4$ \\
$1000000$ &$100 $ & $8.929 $ & $8.680 $ & $2.5 $ & $0.3$ \\
\hline
\end{tabular}
\end{center}
\label{sim-r}
\end{table}
We again observe the superiority of our scheme. 

\subsection{Application to Pricing Double Barrier Options under the Stochastic Volatility Models}\label{sec-double}
Fix $K, K' >0$. Let us 
consider a double barrier option knocked out if price process $X$ 
hit either the boundary $K$ or $K+K'$. 
The price of a 
double barrier option with payoff function $f$ 
and barriers $K$ and 
$K+K'$ is given by 
$\bE[f(X_T)I_{\{\tau_{(K,K+K')} >T\}}]$, 
where $\tau_{(K,K+K')}$ is the first exit time of $X$ from $(K,K+K')$. 
In a similar way as the static hedging formula of 
a single barrier option, we obtain a static hedging 
formula 
if the price process satisfies APCS both at $K$ and $K+K'$. 
\begin{prop}[\citet{AI}]\label{propdouble}
If $X$ satisfies APCS at both $ K $ and 
$ K+K'$, 
then for any bounded Borel function $ f $ 
and $ T > 0 $, we have 
\begin{equation}\label{double}
\begin{split}
&\bE[f(X_T) I_{\{\tau_{(K,K+K')} >T\}}]\\
&=\sum_{n\in \mathbf{Z}} \bE[f({X}_T-2nK')I_{[K + 2n K', K + 
(2n+1) K')}(X_T)]\\
& \quad -\sum_{n\in \mathbf{Z}} \bE f(2K-({X}_T-2nK'))
I_{[ K + (2n-1) K' K + 2n K')}(X_T)], 
\end{split}
\end{equation}
\end{prop}
Of the formula (\ref{double}), the left-hand-side is 
the price of a barrier option, and the 
right-hand-side is 
an infinite series of the prices of plain-vanilla options. 
It means that a double barrier option can be hedged by 
infinite plain-vanilla options. Practically, the series should be 
approximated by finite terms. 
In our numerical scheme, however, 
finite sum approximation is not necessary as we will explain later 
in Remark \ref{rem-45}.

We give a numerical scheme 
of a double barrier option under stochastic volatility model 
by using the symmetrized process which 
satisfies APCS at both $K$ and $K+K'$. 
The scheme is summarized as 
\begin{prop}\label{doublecor}
Set 
\begin{equation}\label{hatsigma}
\begin{split}
&\hat{\sigma}_{11} (x,v) \\
&=  
\sum_{n \in \mathbf{Z}}
\sigma_{11} (x-2nK',v)I_{[K + 2n K', K + 
(2n+1) K')}(x)\\
& \quad
-\sum_{n \in \mathbf{Z}}
\sigma_{11} (2K-(x- 2nK'),v)I_{ [K + (2n-1) K', 
K + 2n K')}(x),
\end{split}
\end{equation}
\begin{equation}\label{hatmu}
\begin{split}
&\hat{\mu}_{1} (x,v) \\
&=  
\sum_{n \in \mathbf{Z}}
\mu_{1} (x-2nK',v)I_{[K + 2n K', K + 
(2n+1) K')}(x)\\
& \quad
-\sum_{n \in \mathbf{Z}}
\mu_{1} (2K-(x- 2nK') ,v)I_{[ K + (2n-1) K', 
K + 2n K')}(x),
\end{split}
\end{equation}
and let $ \hat{X} $ be the unique (weak) solution to 
\begin{equation*}
d\hat{X}_t = \hat{\sigma}_{11} (\hat{X}_t,V_t) dW_t  
+ \hat{\mu}_1 (\hat{X}_t,V_t) \,dt,
\end{equation*}
where $ V $ is the solution to SDE (\ref{GSV1}). 
Then, it holds for any bounded Borel function $ f $ and $ t>0 $ that
\begin{equation}\label{STH7}
\begin{split}
& E [ f( X_t)1_{\{ \tau_{(K,K+K')} > t\}}]  \\
& \quad = 
\sum_{n \in \mathbf{Z} }
E[f(\hat{X}_t- 2nK')I_{ [K + 2n K', 
K + (2n+1) K')}(\hat{X}_t)] \\
& \qquad - \sum_{n \in \mathbf{Z} }
E [f (2K-(\hat{X}_t-2nK'))I_{ [K + (2n-1) K', 
K + 2n K')}(\hat{X}_t)].
\end{split}
\end{equation}
\end{prop}
\begin{proof}
This is an easy consequence of Proposition \ref{propdouble}. 
\end{proof}

\begin{remark}\label{rem-45}
The infinite series of the right hand side 
in (\ref{hatsigma}) and (\ref{hatmu}) is 
expressed by the following;
\begin{eqnarray*}
&&(\mbox{the right hand side of }(\ref{hatsigma}))\\
&&=\left\{
\begin{array}{ll}
\sigma(x- [ \frac{x-K}{K'}]K',v) 
&\mbox{if\ } [ \frac{x-K}{K'}] \equiv 0\mod 2, \\
-\sigma(2K -(x-([ \frac{x-K}{K'}]-1)K' ),v)
&\mbox{if\ }[ \frac{x-K}{K'}] \equiv 1 \mod 2, 
\end{array} 
\right.
\end{eqnarray*}
and 
\begin{eqnarray*}
&&(\mbox{the right hand side of }(\ref{hatmu}))\\
&&=\left\{
\begin{array}{ll}
\mu(x- [ \frac{x-K}{K'}]K',v) & \mbox{if\ }[ \frac{x-K}{K'}] \equiv 0\mod 2, \\
-\mu(2K -(x-([ \frac{x-K}{K'}]-1)K') ,v)
&\mbox{if\ }[ \frac{x-K}{K'}] \equiv 1 \mod 2.
\end{array} 
\right.
\end{eqnarray*}
Therefore the discretized process of 
$(\hat{X}, V) $ by Euler-Maruyama scheme can be 
simulated without approximating the infinite series by
finite sums. 
Similarly, we have 
\begin{eqnarray*}
\begin{split}
&(\mbox{the right hand side of }(\ref{STH7}))\\
&=\bE\left[f(\hat{X}_t[ \frac{[\hat{X}_t-K}{K'}]K',V_t)
I_{\{\frac{x-K}{K'}] \equiv 0\mod 2\}} 
\right. \\
& \qquad \left.
- f(2K -(\hat{X}_t-([ \frac{\hat{X}_t-K}{K'}]-1)K') ,V_t)
I_{\{[\frac{x-K}{K'}] \equiv 1 \mod 2\}}\right].
\end{split}
\end{eqnarray*}
Therefore Put-Call symmetry method is available 
for the pricing of a barrier option. 
\end{remark}
Table \ref{d-heston} and Table \ref{sim-double}
 below are numerical results 
of the pricing of a double barrier call option under 
Heston model and $\lambda$-SABR model, respectively. 
We take 
\begin{enumerate}
\item[Table] \ref{d-heston}: $X_0=100,\ V_0=0.03,\ S=95,\ K+K'=115,\ K=85,\ \theta = 0.03,\ r = 0.02,\ 
T=1,\ \kappa =1,\ \rho=-0.7,\ \nu = 0.03$,
\item[Table] \ref{sim-double}:$X_0=100,\ V_0= 0.3,\ 
S=95,\ K+K'=110,\ K=90,\ \theta = 0.3,\ 
r = 0.02,\ T=1,\ \beta=0.75, \lambda =1,\ 
\rho=-0.7,\ \nu = 0.3$,
\end{enumerate}
in the experiments. 
Benchmark is given by the same setting of Section \ref{cevsim}. 

\begin{table}[htbp]
\caption{
Heston model
\protect\\ \footnotesize{
$X_0=100,\ V_0=0.03,\ S=95,\ K+K'=115,\ K=85,\ \theta = 0.03,\ r = 0.02,\ 
T=1,\ \kappa =1,\ \rho=-0.7,\ \nu = 0.03, \ 
\mbox{benchmark\ of\ option \ price }  = 1.40319930$}}
\begin{center}
\begin{tabular}{c c|c c|c c}
\hline
 No. of & No. of &  &&EM & PCM \\ 
simulation trials & time steps &EM & PCM & 
error(\%) & error(\%) \\ \hline
$1000$ &$10 $ & $2.987 $ & $1.671 $ & $112.869 $ & $19.1 $\\
 $8000$ &$20 $ & $2.368 $ & $1.498 $ & $68.759 $ & $6.7 $\\
 $27000$ &$30 $ & $2.144 $ & $1.588 $ & $52.785 $ & $13.2 $\\
 $64000$ &$40 $ & $2.045 $ & $1.475 $ & $45.770 $ & $5.1 $\\
 $125000$ &$50 $ & $1.921 $ & $1.402 $ & $36.903 $ & $0.1 $\\
 $216000$ &$60 $ & $1.876 $ & $1.453 $ & $33.662 $ & $3.6 $\\
 $343000$ &$70 $ & $1.820 $ & $1.411 $ & $29.728 $ & $0.6 $\\
 $512000$ &$80 $ & $1.792 $ & $1.438 $ & $27.733 $ & $2.5 $\\
 $729000$ &$90 $ & $1.765 $ & $1.411 $ & $25.791 $ & $0.6 $\\
 $1000000$ &$100 $ & $1.744 $ & $1.416 $ & $24.281 $ & $0.9 $\\
\hline
\end{tabular}
\end{center}
\label{d-heston}
\caption{
$\lambda$-SABR model 
\protect\\ 
\footnotesize{
$X_0=100,\ V_0= 0.3,\ 
S=95,\ K+K'=110,\ K=90,\ \theta = 0.3,\ 
r = 0.02,\ T=1,\ \beta=0.75, \lambda =1,\ 
\rho=-0.7,\ \nu = 0.3,\ 
\mbox{benchmark\ of\ option \ price }  = 2.46950606$}}
\begin{center}
\begin{tabular}{c c|c c|c c}
\hline
 No. of & No. of &  &&EM & PCM \\ 
simulation trials & time steps &EM & PCM & 
error(\%) & error(\%) \\ \hline
$1000$ &$10 $ & $3.779 $ & $2.451 $ & $53.017 $ & $0.8 $\\
 $8000$ &$20 $ & $3.427 $ & $2.566 $ & $38.768 $ & $3.9 $\\
 $27000$ &$30 $ & $3.164 $ & $2.442 $ & $28.126 $ & $1.1 $\\
 $64000$ &$40 $ & $3.037 $ & $2.489 $ & $22.997 $ & $0.8 $\\
 $125000$ &$50 $ & $2.955 $ & $2.514 $ & $19.640 $ & $1.8 $\\
 $216000$ &$60 $ & $2.915 $ & $2.480 $ & $18.036 $ & $0.4 $\\
 $343000$ &$70 $ & $2.875 $ & $2.481 $ & $16.438 $ & $0.5 $\\
 $512000$ &$80 $ & $2.838 $ & $2.478 $ & $14.906 $ & $0.4 $\\
 $729000$ &$90 $ & $2.806 $ & $2.464 $ & $13.631 $ & $0.2 $\\
 $1000000$ &$100 $ & $2.779 $ & $2.465 $ & $12.540 $ & $0.2 $\\
\hline
\end{tabular}
\end{center}
\label{sim-double}
\end{table}
We still see that the put-call symmetry method beats 
the path-wise EM. 

\section{Concluding Remark}
The new scheme, which is based on the symmetrization of 
diffusion process, is, though not theoretically, 
experimentally proven 
to be more effective than the path-wise Euler-Maruyama
approximation scheme. 
The scheme is also applicable to stochastic volatility models including 
Heston's and SABR type. 

\bibliographystyle{econometrica}

\end{document}